\documentclass[12pt]{article}

\usepackage{bm}
\usepackage{subfigure}
\usepackage{graphics}
\usepackage{epsfig}
\usepackage{amsfonts}
\usepackage{amsmath}
\usepackage{dsfont}
\usepackage{pifont}
\usepackage{bbm}
\usepackage{multirow}
\usepackage{latexsym}
\usepackage{verbatim}
\usepackage{arydshln}
\usepackage{cancel}
\usepackage{slashed}
\usepackage{natbib}
\usepackage{amsmath,epsfig,amssymb,amsfonts,amsthm,epsfig,verbatim}

\def\bbbr{{\rm I\!R}}

\def\Skip{\par\bigskip\nobreak}

\newtheorem{proposition}{Proposition}
\def\Skip{\par\bigskip\nobreak}
\def\dj{d\kern-0.4em\char"16\kern-0.1em}
\def\Dj{\mbox{\raise0.3ex\hbox{-}\kern-0.4em D}}

\DeclareMathAlphabet{\mathpzc}{OT1}{pzc}{m}{it}

\begin{document}
\pagestyle{plain}

\makeatletter
\@addtoreset{equation}{section}
\makeatother
\renewcommand{\thesection}{\arabic{section}}
\renewcommand{\theequation}{\thesection.\arabic{equation}}
\renewcommand{\thefootnote}{\arabic{footnote}}

\setcounter{page}{1}
\setcounter{footnote}{0}

\begin{titlepage}
\begin{flushright}
\small ~~
\end{flushright}

\bigskip

\begin{center}

\vskip 0cm

{\LARGE \bf   { Asymmetric latent semantic indexing for gene expression experiments visualization}} \\[6mm]

\vskip 0.5cm

{\bf Javier Gonz\'alez$^1$,\, Alberto Mu\~noz$^2$  \,and\, Gabriel Martos$^2$}\\

\vskip 25pt

{\em
              $^1$Sheffield Institute for Translational Neuroscience, \\Department of Computer Science, University of Sheffield.\\
               Glossop Road S10 2HQ, Sheffield, UK.\\
}
{\small {\tt \ { j.h.gonzalez@sheffield.ac.uk }}}
\Skip

{\em $^2$Department of Statistics, University Carlos III of Madrid \\ Spain.  C/ Madrid, 126 - 28903, Getafe (Madrid), Spain.}\\
{\small {\tt \ {alberto.munoz@uc3m.es, gabriel.martos@uc3m.es }}}

\vskip 0.8cm

\end{center}

\vskip 1cm

\begin{center}

{\bf ABSTRACT}\\[3ex]

\begin{minipage}{13cm}
\small

We propose a new method to visualize gene expression experiments inspired by the latent semantic indexing, technique originally proposed in the textual analysis context. By using the correspondence word-gene document-experiment, we define an asymmetric similarity measure of association for genes that accounts for potential hierarchies in the data, the key to obtain meaningful gene mappings. We use the polar decomposition to obtain the sources of asymmetry of the similarity matrix, which are later combined with previous knowledge. Genetic classes of genes are identified by means of a mixture model applied in the genes latent space.  We describe the steps of the procedure and we show its utility in the Human Cancer dataset.  

\vspace{0.5cm}

\end{minipage}

\end{center}

\vfill

\end{titlepage}

\section{Introducction}
A gene expression dataset consists of a matrix $\textbf{Y} \in \bbbr^{n \times p}$, with each row representing an experiment and each column representing a gene. Typically, the number of genes is  several thousand, whereas the number of experiments or samples is in the order of tens. In Figure \ref{figure:heatmap}.A we show the heat map of the differentially expressed genes of the Human cancer dataset, which originally consists of  6830 genes measured in 64 experiments corresponding to 14 different types of Cancer patients available in \cite{bib:Hastie2009}. To provide answers to questions like which genes are more similar in terms of their expression profiles or which genes are involved in certain types of cancer is the key to extracting useful biological knowledge in experiments of this type. 

A common strategy to find interesting patterns in the data is to define some measure of similarity or dissimilarity for the genes \citep{journals/bmcbi/PrinessMB07,citeulike:5767111}, which is later combined with a cluster algorithm \citep{Kohonen:2001:SM:558021,journals/bioinformatics/Gat-ViksSS03}. The Euclidean distance, the Pearson correlation coefficient or the Mutual Information, are the most common measures. Although useful in many scenarios, such measures are unable to capture some complex features that have been discovered to be present in the way the genes interact with each other. Particularly, an interesting case is the hierarchy among the genes, an universal pattern that has been extensively observed in the literature, mainly in the context of networks analysis \citep{Reka:RevModPhys2002, wuchty03architecture,citeulike:400265}.

Inspired by the latent semantic indexing (LSI) \citep{deerwester-88,Deerwester90indexingby}, the technique originally proposed for textual data analysis, in this paper we propose a new visualization technique to unravel the structure of gene expression datasets.  Although the idea of using textual data analysis techniques in the biological context has been explored in the literature in some recent works \citep{Bicego:2010:EMC:1774088.1774415,Ng:2004:WFC:976520.976537,Caldas2009}, these approaches use the Latent Dirichlet Allocation (LDA) \citep{Blei03latentdirichlet} as a fundamental model, which  provides neither a Euclidean representation of the genes useful for visualization nor takes into account the hierarchical relationship among the genes. In this work we address both problems by means of a new asymmetric latent semantic indexing approach (aLSI), following the existing literature in asymmetric similarities based methods \citep{okada1987nonmetric,okada1990generalization,chino1978graphical,chino1990generalized,munoz2003support}. Therefore, the contributions of this paper are twofold: 

\begin{itemize}
\item[(i)] A proof-of-concept analysis  to  illustrate the importance of using  asymmetric gene similarities in gene expression experiments.
\item[(ii)] A new asymmetric latent semantic indexing (aLSI) approach to produce meaningful gene mappings, which can be used in combination with previous biological knowledge such as gene-ontologies, pathways, protein-protein interaction networks, etc.
\end{itemize}

Our approach is inspired by the work of  \citep{AlbertoMu&ntilde;oz2012} in which an asymmetric version of the LSI is already  defined in the textual data context. In this work the authors propose a partition of the data in several hierarchical levels, which  aim at accounting for the hierarchical relationships between the words of the database. Within each level, a Gram Mercer kernel matrix is obtained by means of the triangular decomposition, which captures the remaining asymmetries not removed by the partition in the different layers. Finally, a Euclidean representation of the words is produced within each level and these are connected using a measure of inclusion. 

In this work, we propose an alternative aLSI which does not require a partition of the dataset in hierarchical levels. This represent itself and advantage with respect to the work of  \citep{AlbertoMu&ntilde;oz2012} since the choice of the number of layers its already avoided.  Nevertheless, the key aspect of our approach is to replace the triangular decomposition of the similarity matrix by the polar decomposition, which produces two complementary gene representations. This allows us to produce a global mapping that does not require any partition of the data while the information provided by the asymmetries in the gene similarity matrix is still taken into account. 

\begin{figure}[t!]
 \begin{center}
\includegraphics[height=12.5cm,]{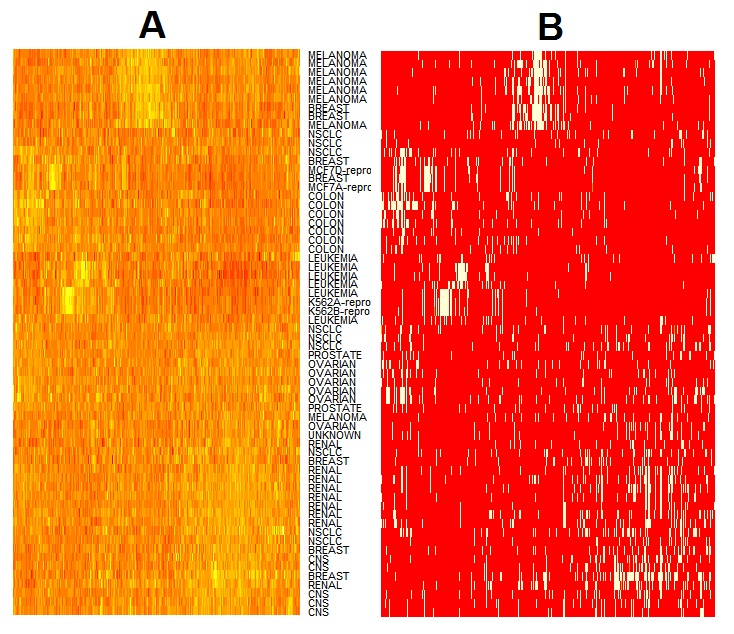}
\caption{A) Heat-map of the micro-array of the Human Cancer dataset. Originally, there are 6830 genes (columns) whose expression is measured in 64 patients (rows) with 14 different types of Cancer. Colour intensity represents the level expression of the genes. B) Heap map of the Human Cancer dataset in which only the expressed genes are highlighted (in white). Each row of this matrix can be interpreted as a document whose words are those genes which are differentially expressed.  }\label{figure:heatmap}
  \end{center}
\end{figure}

This paper is organized as follows. In Section \ref{sec2} we detail the connection between asymmetric similarities and hierarchies in genetic experiments and we illustrate this phenomenon in the Human Cancer data set. In Section \ref{sec3} we propose a new asymmetric latent semantic indexing (aLSI) procedure. In Section \ref{sec4} we illustrate the utility of the proposed approach in a real data experiment and in Section \ref{sec5} we conclude with a discussion of this work.


\begin{figure}[t!]
 \begin{center}
\includegraphics[height=10cm]{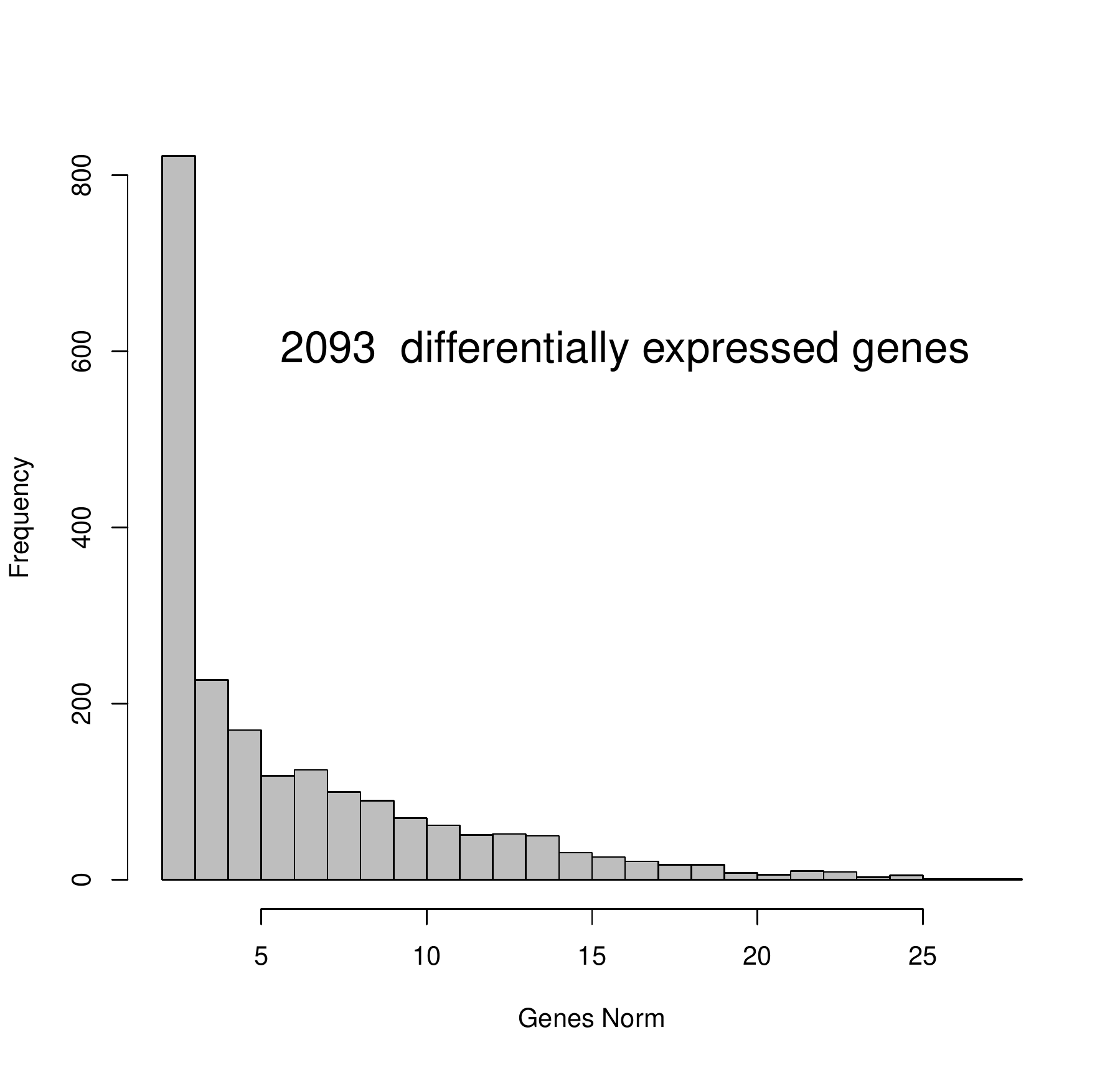}
\caption{Evidence of the Zipf's law in gene expression experiments: Histogram of the norms of the 2093 differentially expressed genes of the Human cancer data set. }\label{figure:zipf}
  \end{center}
\end{figure}

\section{Hierarchy/asymmetry in gene expression experiments}\label{sec2}

In this section we illustrate the idea of ``gene hierarchy". To this end, we will use the above mentioned Human Cancer data set. Consider the matrix $\textbf{X}$  such that $\textbf{x}_{kj}=1$ if the gene $j$ is significantly expressed in the experiment $k$ and $\textbf{x}_{kj}=0$ otherwise (see Section \ref{sec41} for details). This gene-experiment matrix is analogous to the  term-document matrix, common in textual analysis \citep{AlbertoMu&ntilde;oz2012}. In this field, it is common to work with a matrix $\textbf{X}$ where $\textbf{x}_{kj}=1$ if the term $j$ appears in document $k$ and $\textbf{x}_{kj}=0$ otherwise. By using the correspondence genes/words and experiments/documents we can apply techniques from the text mining literature to analyse gene expression datasets. Therefore, in the sequel we will use indistinctly the terms genes-words and experiments-documents.

For now, consider a textual data set and let $|\textbf{x}_i|$ be  the number of documents indexed by term $ith$ and  $|\textbf{x}_i \wedge \textbf{x}_j|$ the number of documents indexed by both $i$ and $j$ terms. Consider the following asymmetric similarity measure ($s_{ij} \neq s_{ji}$)

\begin{equation}\label{eq:similarity}
s_{ij} = \frac{|\textbf{x}_i \wedge \textbf{x}_j|}{|\textbf{x}_i|} = \frac{\sum_k  min (x_{ik},x_{jk}) }{\sum_k x_{ik} },
\end{equation}
which has been previously studied in a number
of works related to Information Retrieval \citep{citeulike:461218, bib:diego10}. It turns out that expression (\ref{eq:similarity}) can be interpreted as the degree in which the topic represented by the
term $i$ is a subset of the topic represented by the term $j$. As a measure of
inclusion it was originally proposed by \cite{bib:Kosko1991} in the context of fuzzy set theory. Regarding its interpretation in a textual data example, consider, for instance, a collection of documents containing the term ``statistics". In this case a more specific term like
``non parametric" will occur just in a subset. The relation between ``non parametric"
and "statistics" is strongly asymmetric, in the sense that the concept represented
by the word ``non parametric" is a subset of the concept represented by the word
``statistics" but not conversely. In the biological context, where $s_{ij}$ represents the similarity between two genes, expression (\ref{eq:similarity}) represents the degree in which a gene $i$ is a subclass or it is hierarchically dependent of a gene $j$. 

The matrix $\textbf{X}$ contains information about both, the terms and the documents of the database. In the sequel we will use $\textbf{t}_j$ to refer the terms (columns of $\textbf{X}$) and $\textbf{d}_i$ (rows of \textbf{X}) to refer the documents. Using the definition of similarity in expression (\ref{eq:similarity}) the skew-symmetric term associated to each pair of terms $\textbf{t}_i$, $\textbf{t}_j$ can be written as 
$$\frac{1}{2}(s_{ij}-s_{ji}) = \frac{1}{2}\left(\frac{|\textbf{t}_i\wedge \textbf{t}_j|}{|\textbf{t}_i|}-\frac{|\textbf{t}_i\wedge \textbf{t}_j|}{|\textbf{t}_j|} \right)= \frac{|\textbf{t}_i\wedge \textbf{t}_j|}{2|\textbf{t}_i||\textbf{t}_j|}(|\textbf{t}_j|-|\textbf{t}_i|)\propto (|\textbf{t}_j|-|\textbf{t}_i|).$$
Therefore, a large difference between $s_{ij}$ and $s_{ji}$ is directly related to a large difference between the norms of the words given by
$|\textbf{t}_i|$ and $|\textbf{t}_j|$. Thus, the distribution of term norms in case of asymmetry/hierarchy is clearly far from being uniform. 

In Figure \ref{figure:zipf} we show the histogram of the norms of the differentially expressed genes of the Human Cancer data set. The figure shows that a few number of genes have very large norms while a large number of genes have small norms. This behaviour, which  can be modelled by means of the Zipf's law \citep{journals/ijon/Martin-MerinoM05}, is an evidence of asymmetric/hierarchical associations. Genes with large norms correspond to `biologically relevant' genes involved in many processes (or high level concepts), whereas genes with small norm represent rarely expressed genes (or very specific concepts). The hierarchy induced on the gene set by the inclusion measure $s_{ij}$ is directly related with its asymmetric nature, and caused by the strongly asymmetric gene frequency distribution.



\section{Asymmetric latent semantic indexing}\label{sec3}

The latent semantic indexing (LSI) \citep{deerwester-88} is a useful technique in natural language processing to analyse relationships between a set of documents and the terms they contain. The idea is to produce a set of concepts or latent semantic classes to summarize the content of the dataset. In this section we propose an asymmetric latent semantic indexing that uses as input the similarity in eq. (\ref{eq:similarity}). In a biological context, we will talk about `latent genetic classes' to refer to groups of genes that summarize the main content of the data. Next, we introduce the LSI to later generalize it to its asymmetric version.

\subsection{Latent semantic indexing}

Consider the $n\times p$ document by term $\textbf{X}$ matrix whose entries contain the word counts per document. The  matrix $\textbf{X}^T \textbf{X}$ contains the correlations among terms $\textbf{t}_j$ and $\textbf{t}_k$ (measured as  $\textbf{t}_j^T \textbf{t}_k$) and $\textbf{X} \textbf{X}^T$ contains the correlations among documents measured as $\textbf{d}_i \textbf{d}_s^T$. Using the singular value decomposition (SVD) for $\textbf{X}$ we obtain the unique decomposition $\textbf{X} = \textbf{U}_x \Sigma_x \textbf{V}_x^T$, where $\textbf{U}_x$ and $\textbf{V}_x$ are orthogonal matrices and $\textbf{D}_x$ is diagonal and contains the singular values of $\textbf{X}$. It is straightforward to see that $\textbf{X} \textbf{X}^T = \textbf{U}_x \Sigma_x \Sigma_x^T \textbf{U}_x^T$ and

\begin{equation}
\textbf{X}^T \textbf{X} = \textbf{V}_x \Sigma_x^T \Sigma_x \textbf{V}_x^T.
\end{equation}
Therefore, the immersion of the term $\textbf{t}_j$ into the semantic class space is given by 
\begin{equation}
\textbf{t}_j=\Sigma^{-1}_x \textbf{U}_x \textbf{t}_j.
\end{equation}
On the other hand, the immersion of document $\textbf{d}_i$ in the same latent space is given by $\textbf{d}_i = \Sigma^{-1}_x \textbf{V}_x \textbf{d}_i$.

\subsection{Polar decomposition of an asymmetric similarity matrix}\label{sec32}
Consider the $p \times p$ asymmetric similarity matrix $(\textbf{S})_{ij} = s_{ij}$  in eq. (\ref{eq:similarity}). By means of the SVD we obtain that $\textbf{S} = \textbf{U}_s \Sigma_s \textbf{V}_s^T$, which lead to the polar decomposition of $\textbf{S}$ \citep{bib:Horn19991, bib:Higham1986}.  Define $\textbf{L}=\textbf{U}_s\textbf{V}^T_s$. Then $\textbf{S} = \textbf{K}_1\textbf{L} = \textbf{L}\textbf{K}_2,$ where 
\begin{equation}\label{k1}
\textbf{K}_1 = \textbf{U}_s \Sigma_s \textbf{U}_s^T 
\end{equation}
\begin{equation}\label{k2}
\textbf{K}_2 = \textbf{V}_s\Sigma_s \textbf{V}_s^T. 
\end{equation}
Note that $\|\textbf{S}\|_F = \|\textbf{K}_1\|_F = \|\textbf{K}_2\|_F$, where $\|\cdot\|_F$ is the Frobenius norm. Also remark that $\textbf{S}$ does not directly decompose in any combination of $\textbf{K}_1$ and $ \textbf{K}_2$ but these matrices can be understood as the two sources of asymmetry of $\textbf{S}$.  Geometrically speaking, since $\textbf{S}\textbf{V} =   \textbf{U} \Sigma$, it is straightforward to check that  $\textbf{S}\textbf{v}_j =   \sigma_j\textbf{u}_j$ where $\textbf{v}_j$ and $\textbf{u}_j$ are the columns of $\textbf{U}$ and $\textbf{V}$ respectively. Therefore the eigenvectors $\{\textbf{v}_1,\dots, \textbf{v}_p \}$ of $\textbf{K}_2$ are mapped under the asymmetric matrix $\textbf{S}$ onto the scaled orthogonal coordinate system $\{\sigma_1 \textbf{u}_1,\dots,\sigma_n \textbf{u}_p \}$. Equivalently, one can interpret the symmetric effect with respect to the eigenvectors $\{\textbf{u}_1,\dots, \textbf{u}_p \}$. The asymmetry in $\textbf{S}$ is therefore reflected in the angle between each pair of left and right eigenvectors of $\textbf{S}$. Therefore $span\{\textbf{v}_1,\dots,\textbf{v}_p\}$ and $span\{\textbf{u}_1,\dots,\textbf{u}_p\}$ produce different but complementary representations of the genes.  Note that if $\textbf{S}$ is a symmetric matrix $\textbf{K}_1 =  \textbf{K}_2$ and therefore both representations are equivalent since $\textbf{u}_j = \textbf{v}_j$ for all $j=1,\dots,p$. The polar decomposition has been previously used in the analysis of asymmetric relationships in \citep{Gower1977, Gower1998}.

\subsection{Merging the sources of asymmetry}\label{sec:comb}

The matrices $\textbf{K}_1$ and $\textbf{K}_2$ are symmetric and positive semi-definite. Therefore, they are kernel matrices \citep{bib:aroszajn50,bib:wahba90} that admit the decompositions $\textbf{K}_1 = \Phi_1\Phi_1^T$ and $\textbf{K}_2 = \Phi_2\Phi_2^T$ where $\Phi_1= \textbf{U}\Sigma^{1/2}$ and $\Phi_2= \textbf{V}\Sigma^{1/2}$ respectively. The two matrices induce two different distances for the terms, which are the consequence of $\textbf{S}$ of being asymmetric. Note that if $\textbf{S}$ is symmetric then $\Phi_1=\Phi_2$. To find a unifying distance (or kernel) using $\textbf{K}_1$ and $\textbf{K}_2$  is therefore the key to obtain an appropriate Euclidean representation for the terms. In this sense,  suppose that we are able to find suitable transformations $\phi_i$, $i=1, 2$, such that the induced distance on the terms, given by $d_{\phi_i}(\textbf{t}_j,\textbf{t}_k)^2 = \|\phi_i(\textbf{t}_j) - \phi_i(\textbf{t}_k)\|^2$, corresponds to the one induced by each kernel matrix $\textbf{K}_i$. This implies that $d_{\phi_i}(\textbf{t}_j,\textbf{t}_k)^2 = (\textbf{K}_i)_{jj} + (\textbf{K}_i)_{kk} - 2 (\textbf{K}_i)_{jk}$, where $j,k =1,\dots,p$ and $(\textbf{K})_{jk}=\phi_i(\textbf{t}_j)^T\phi_i(\textbf{t}_k)$.

Following  \citep{RePEc:eee:jmvana:v:120:y:2013:i:c:p:120-134}, it is possible to prove that for each matrix $\textbf{K}_i$ there exists a symmetry, continuous and positive-definite kernel function $K_i: T\times T \rightarrow \bbbr$, where $T$ is a compact set, such that $K_i(\textbf{t}_j,\textbf{t}_k) = \phi(\textbf{t}_j)^T\phi(\textbf{t}_k)$, $\textbf{t}_j$, $\textbf{t}_k\in T$ is the implicit kernel corresponding to $d_{\phi_i}(\textbf{t}_j,\textbf{t}_k)$. See \citep{RePEc:eee:jmvana:v:120:y:2013:i:c:p:120-134} for conditions on the existence of such $k_i$. Each kernel function $k_i$ has a unique associated Reproducing kernel Hilbert space (RKHS), whose feature map \footnote{We say that $\phi$ is the feature map of a kernel $k:T\times T \rightarrow \bbbr$ if $k(\textbf{t},\textbf{t}') = \langle \phi(\textbf{t}),\phi(\textbf{t}') \rangle$ holds for any $\textbf{t}$, $\textbf{t}' \in T$ where  $\langle \cdot,\cdot \rangle$ represents the usual $l^2$ product.} or canonical basis, is given by $\phi_i$  \citep{bib:aroszajn50,bib:wahba90}. 

The operation of adding the kernels $k_1$ and $k_2$ gives rise to a new RKHS whose feature map is the union of $\phi_1$ and $\phi_2$. In particular, let $k_1$ and $k_2$ two positive semi-definite kernel functions and let $\phi_1$ and $\phi_2$ their underlying feature maps. Then $k=\lambda_1k_1+\lambda_2k_2$, with $\lambda_1,\,\lambda_2\geq 0$, is a positive semi-definite kernel with $\phi = [\sqrt{\lambda_1}\phi_1,\sqrt{\lambda_2}\phi_2]$ as a valid feature map. This property, , which can be easily generalized to multiple kernels, implies that the sum of the kernel functions $k_1$ and $k_2$ can be understood as the sum of the associated RKHSs. Therefore, to use the operation $\textbf{K}=\lambda_1\textbf{K}_1+\lambda_2\textbf{K}_2$, with $\lambda_1=\lambda_2=1/2$, has the property of defining a new kernel matrix whose induced distances take  equally into account the representation of the terms using both kernels, or equivalently in our case,  the representations of the genes given by $\Phi_1= \textbf{U}\Sigma^{1/2}$ and $\Phi_2= \textbf{V}\Sigma^{1/2}$. That is, the right and left eigenvalues of $\textbf{S}$ have the same weight on the final distance induced by $\textbf{K}$. 

An alternative fusion scheme can be found in \citep{AlbertoMu&ntilde;oz2012}. However, in this work the main step to deal with asymmetry is to split the dataset into layers of words with similar norm. Here, we are able to deal with asymmetry in a single step by  means of the polar decomposition of $\textbf{S}$. In the former approach,  hierarchical clusters of words are provided, but a unique representation of the terms is not available as we provide here. This represents a problem for the generalization and applicability of the work in \citep{AlbertoMu&ntilde;oz2012} that is solved in our proposal: since the distance among words of different layers is not available, this technique cannot be used in problem like classification in which a unique distance for the words is needed.

\subsection{Generalizing the combination approach}

The goal of this section is  to generalize the previous idea described in the previous section in order to propose an approach to combine $\textbf{K}_1$, $\textbf{K}_2$ and a third matrix $\textbf{W}$ with prior information about the problem. Such a matrix might be derived from an initial labeling of the terms or the experiments. In the genetic context, this is a natural idea since prior knowledge about the relationships among the genes is common \citep{wang2013incorporating}. Some examples are gene-ontologies, pathways, protein-protein interaction networks, etc. Note that by imposing $\textbf{K}$ to be positive semi-definite a Euclidean representation of the terms is always available by mean of some matrix decomposition $\textbf{K}=\Phi\Phi^T$ \citep{Schoenberg35remarksto,Householder1938}.

We combine $\textbf{K}_1$, $\textbf{K}_2$ and $\textbf{W}$ to obtain a fusion similarity matrix $\textbf{K}$ by maximizing
\begin{equation}\label{eq:optSpenalized}
\textsf{G}_{\tau} (\textbf{K}) = \left \|\textbf{K} - \gamma_1\mathcal{F}(\textbf{K}_1,\textbf{K}_2)\right \|^2_F + \tau \left \|\textbf{K} - \gamma_2\textbf{W} \right \|_F^2, 
\end{equation}
where $\tau> 0$ is the regularization parameter, $\gamma_1,\gamma_2>0$ are scale parameters and $\mathcal{F}(\textbf{K}_1,\textbf{K}_2)$ is a functional combination of the matrices $\textbf{K}_1$ and $\textbf{K}_2$ whose output is a symmetric positive semi-definite matrix.  The underlying idea in eq. (\ref{eq:makm}) is to merge both sources of asymmetry and to keep a balance with the prior knowledge given by $\textbf{W}$. The fusion scheme proposed in eq. (\ref{eq:makm}) can be derived using a regularization theory approach, similar to the one used in the derivation of SVM classifiers \citep{bib:diego10}. The solution to the problem stated  in eq. (\ref{eq:optSpenalized}) is given in the following proposition,

\begin{proposition}
The minimizer of $\textsf{G}_{\tau} (\textbf{S})$ for any $\mathcal{F}$ and $\tau>0$ and $\gamma_1=\gamma_2 = \tau+1$ is given by

\begin{equation}\label{eq:makm}
\textbf{K}= \mathcal{F}(\textbf{K}_1,\textbf{K}_2) +\tau \textbf{W}.
\end{equation} 
\end{proposition}
Of course, different $\mathcal{F}$ lead to different combinations of $\textbf{K}_1$ and $\textbf{K}_2$. In this work, and based on the ideas described in the previous section, we consider the arithmetic mean of the matrices $\mathcal{F}(\textbf{K}_1,\textbf{K}_2) = (\textbf{K}_1+\textbf{K}_2)/2$  but we refer to  \citep{bib:diego10,Munoz:2007:JDK:1782914.1782982,conf/ciarp/MunozG08,conf/ciarp/MunozGD06} for further kernel fusion procedures.

\subsection{Probabilistic latent semantic indexing with asymmetric similarities}\label{sec:plsi}
In this section we make use of eq. (\ref{eq:optSpenalized}), computed from the asymmetric similarity matrix $\textbf{S}$, to redefine the LSI. We use the ideas from \citep{DBLP:conf/cikm/ParkR09, AlbertoMu&ntilde;oz2012} with the special novelty that the term representation is given by the distances induced by our particular choice of $\textbf{K}$. 

Following the ideas described in Section \ref{sec:comb}, let $\phi$ be a transformation of the terms  such that the induced distance on the terms, given by $d_{\phi}(\textbf{t}_j,\textbf{t}_k)^2 = \|\phi(\textbf{t}_j) - \phi(\textbf{t}_k)\|^2$, corresponds to the one induced by the kernel matrix $\textbf{K}$.  Consider the matrix $\Phi$, such that $(\Phi)_{ij} = \phi_i(\textbf{t}_j)$. The rows of $\Phi$, say the $\phi(\textbf{t}_j)$, represent the transformation of $\textbf{t}_j$ to the latent class/feature space.  Following the LSI scheme, we apply the SVD to the transformed term $p\times m$ matrix $\Phi = \textbf{U} \Sigma \textbf{V}^T$
and we obtain that $\textbf{K} = \Phi \Phi^T = \textbf{U} \Sigma \Sigma^T \textbf{U}^T = (\textbf{U} \Lambda^{\frac{1}{2}})(\textbf{U} \Lambda^{\frac{1}{2}})^T$, where $\Lambda = \Sigma \Sigma^T = \Sigma^2$ is the diagonal matrix of eigenvalues of $\textbf{K}$ and $\Sigma$ is the diagonal matrix of singular values of $\Phi$. In this context the matrix $\textbf{K}$ plays the role of $\textbf{X}^T \textbf{X}$ in the original LSI formulation. Then the immersion of $\phi(\textbf{t}_i)$ is given by 
$$\phi^s(\textbf{t}_i) = \Sigma^{-1} \textbf{U}^T \phi(\textbf{t}_i) = \Lambda^{-\frac{1}{2}} \textbf{U}^T \phi(\textbf{t}_i).$$
Therefore, by replacing $\textbf{X}^T\textbf{X}$ by $\textbf{K}$ we `kernelize' the LSI by using the original asymmetric similarity matrix $\textbf{S}$: we replace the original linear mapping of the LSA by the non linear one give by $\phi$. 

The semantic classes in the latent space can be identified with clusters of transformed term data.  In order to estimate such semantic classes $c_1,\dots,c_q$ we apply a Gaussian mixture model-based clustering \citep{Fraley00model-basedclustering}. That is, for each term we obtain an estimation of the probability of membership, $p(c_i|\textbf{t}_j)$, to each one of the latent semantic classes $c_i$. We assume that each cluster is generated by a Gaussian multivariate distribution $f_k(\textbf{t}) = \mathcal{N}_k(\mu_k ,\Sigma_k )$, where $\mu_k$ and $\Sigma_k$ are the mean vector and covariance matrix respectively. The final mixture density is therefore given by
$$f(\textbf{t}) = \sum_{k=1}^q \alpha_k  \mathcal{N}_k(\textbf{t}) = \sum_{k=1}^q \alpha  \mathcal{N}_k( \mu_k ,\Sigma_k ),$$
where each $\alpha_k$ represents the prior probability or weight of the component $k$. 
The main advantage of this approach is that we can obtain a density estimator for each cluster and a `soft' classification rule is available: each
term may belong to more than one semantic class via the use of conditional probabilities $p(c_i|\textbf{t}_j)$.



\subsection{Algorithm}
In this section we summarize the steps to apply the proposed asymmetric latent semantic indexing to a data set. As we detailed in Section \ref{sec2}, there exist strong similarities between textual and gene expression data, therefore our proposal can be used in both scenarios. See Table \ref{al1} for details.

\begin{table}[t!]
\centering
\begin{tabular}{lll}
\hline
\textbf{Input:} & & Genes-by-experiments matrix $\textbf{X}$. \\
\textbf{Output:} & & Map of terms (genes), latent semantic classes.\\
   \hline
& 1. & Obtain the asymmetric similarity $\textbf{S}$.\\
& 2.& Decompose $\textbf{S}= U_s \Sigma_s V_s^T$.\\
& 3.& Obtain the two sources of asymmetry $\textbf{K}_1$ and $\textbf{K}_2$.\\
& 4.& Obtain the matrix of labels of the terms (or genes) $\textbf{W}$.\\
& 5.& Fuse the matrices using the scheme proposed in (\ref{eq:makm}).\\
& 6.& Obtain the projections of the terms into the latent semantic classes.\\
& 7.& Assign probabilities to the classes using a mixture model.\\
& 8.& Visualize the genes and the mixture model using MDS.\\
   \hline
\end{tabular}\caption{Main steps of aLSI algorithm.}\label{al1}
\end{table}

\section{Application: aLSI of the Human cancer data set}\label{sec4}

In this section we analyse the Human cancer data set, described in the introduction of this work, by using the proposed asymmetric latent semantic indexing detailed in Section \ref{sec:plsi}. The analysis consists of two main steps. First, we calculate the genes which are statistically expressed in each experiment and we obtain the matrix $\textbf{X}$. Second, we use this matrix to obtain genetic semantic classes of genes that we will associate with different types of cancer. In order to find the clusters of genes, we also use the Euclidean distance and the Correlation matrix to illustrate the benefits of our approach in this context. The R-code to replicate all the figures and results of this work is available at https://github.com/javiergonzalezh/aLSI.


\subsection{Differential analysis}\label{sec41}

The initial point in our analysis is the matrix $\textbf{Y}$, which consists of the expression level of 6830 genes in 64 experiments. The first step is to identify which genes are differentially expressed. That is, to statistically decide  whether for a given gene its expression is greater than what we would expect just due to natural random variations. 

\begin{figure}[t!]
 \begin{center}
\includegraphics[height=8cm]{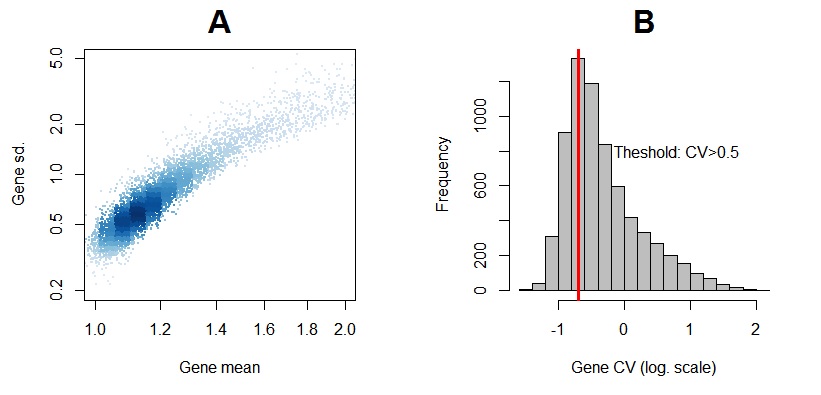}
\caption{The first step towards the identification of the latent genetic classes of the database is to perform a differential analysis of the genes. A) Mean vs. Standard deviation of the  6830 genes of the Human cancer data set across the 64 available patients. B) Histogram of the CV of all the genes.}\label{figure:seriesyoga}
  \end{center}\label{figure.cv}
\end{figure}
The motivation for this gene “filtering” is that a relatively few number of genes of the database should be expressed in each experiment. Different methods have been proposed in the literature \citep{citeulike:12453949}. In this work, we follow a simple and straightforward approach which uses the coefficient of variation $CV=|\bar{x}|/sd(x)$  to discriminate between expressed and non-expressed genes. The reason to use this coefficient is the linear relationship between the gene expression mean and the gene standard deviation expression of the genes across the experiments. See Figure \ref{figure.cv} A. In particular, we consider that a gene is differentially expressed in the database if the value of the coefficient of variation is larger than 0.5. Of course, other thresholds are possible if additional information about the experiential noise is available. In Figure \ref{figure.cv}.A, we show the histogram of the coefficients of variation of all the genes of the database. The total number of genes  with a CV larger than 0.5 is 2093.

Given the set of expressed genes, in order to build the matrix $\textbf{X}$, we need to decide when a particular gene is expressed in an experiment. To this end, we consider the maximum of the expression in the set of non expressed genes and we use it as a threshold in the set of the expressed ones. The purpose of this threshold is to capture the random variation in the data. Figure \ref{figure.exp} shows the expression values of two genes across the 64 experiments. One of the genes (left) is differentially expressed in those experiments above the selected threshold (horizontal dotted line at 4.46). In particular, this gene is assumed to be significantly expressed in a total of 5 experiments. On the other hand, in Figure \ref{figure.exp} (right), we show the expression values of a non expressed gene. All the values remain below the threshold, reflecting that the variations in expression are random variations. In Figure \ref{figure:heatmap}.B, we show the heat map of the 2093 differentially expressed genes of the dataset.

\begin{figure}[t!]
 \begin{center}
\includegraphics[height=8cm,]{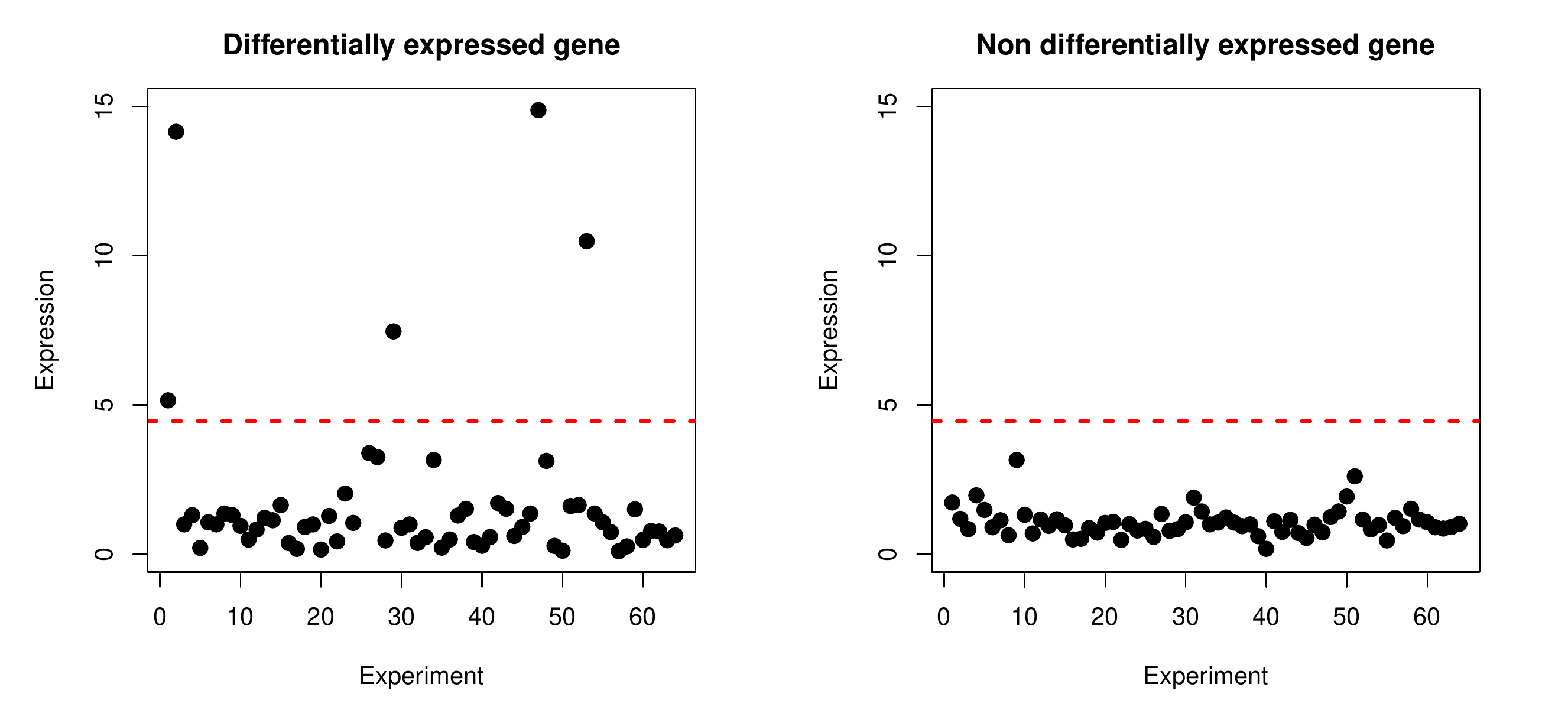}
\caption{Illustration of the profiles of two genes. On the left we show a differentially expressed genes in 5 experiments. On the right we show the profile of a non differentially expressed gene.}\label{figure.exp}
  \end{center}
\end{figure}

\subsection{Extraction of latent genetic classes using aLSI}

Next, we apply the asymmetric latent semantic indexing proposed in Section \ref{sec:plsi} to the differentially expressed genes of the Human Cancer dataset. To this end, we calculate the gene similarity following (\ref{eq:similarity}) and we proceed with the steps of Algorithm 1. 

\begin{table}[ht]
\centering
\begin{tabular}{cccccc}
  \hline
Latent genetic class & Gene 1 & Gene 2 & Gene 3 & Gene 4 & Gene 5 \\ 
  \hline
 1 &   7 &   8 & 619 & 683 & 1726 \\ 
  2 & 1891 & 193 & 187 & 188 & 186 \\ 
  3 & 1721 & 1720 & 1684 & 1620 & 1653 \\ 
  4 &  19 &  59 &  63 &  76 & 102 \\ 
  5 & 1339 & 1619 & 2040 & 1596 & 1470 \\ 
  6 & 1359 & 574 & 1451 & 1729 & 2007 \\ 
  7 & 130 & 156 & 157 & 158 & 249 \\ 
  8 & 496 & 502 & 515 & 493 & 494 \\ 
  9 & 242 & 338 & 369 & 377 & 380 \\ 
  10 & 253 & 277 & 278 & 279 & 281 \\ 
  11 & 996 & 1000 & 1045 & 992 & 1007 \\ 
  12 & 449 & 451 & 475 & 485 & 486 \\ 
  13 & 480 & 576 & 577 & 578 & 1168 \\ 
  14 & 828 & 883 & 884 & 893 & 913 \\ 
   \hline
\end{tabular}\caption{5 genes IDs with maximum probability in the mixture-model for each one of the 14 latent genetic clusters. The label of each gene is given by the row position in the  dataset of differentially expressed genes.}\label{table.10genes}
\end{table}

The matrix $\textbf{W}$ in expression (\ref{eq:makm}) is calculated using the labels of the experiments. First we assign a membership of the genes to each one of the 14 types of cancer:``CNS", ``RENAL'', ``BREAST'', ``NSCLC'', ``UNKNOWN'', ``OVARIAN'',  ``MELANOMA'', ``PROSTATE'', ``LEUKEMIA'', ``K562B-repro'', ``K562A-repro", ``COLON", ``MCF7A-repro'', and ``MCF7D-repro". To this end, we assign the gene $i$ to the type of cancer $k$ if it is expressed in at least in one of the experiments of that type. Note that the same gene might belong to more than one class simultaneously.  We define the gene similarity matrix $\textbf{Q}$ whose entries are calculated as

\begin{equation}
\textbf{q}_{ij} = \frac{\# \mbox{times gene i and j appear simultaneously in some type of cancer}}{\# \mbox{types of cancer in which gene i is expressed}}
\end{equation}
The matrix $\textbf{W}$ in  (\ref{eq:makm}) is calculated as $\textbf{W} = (\textbf{Q}_1+\textbf{Q}_2)/2$ where $\textbf{Q}_1$ and $\textbf{Q}_2$ are the matrices resulting from the polar decomposition of $\textbf{Q}$. Note that the matrix $\textbf{W}$ play the role of the labels in the combination, following the idea of kernel combinations in the support vector classification context \citep{bib:diego10}. Parameter $\tau$ is fixed to $0.2$ following \citep{RePEc:eee:jmvana:v:120:y:2013:i:c:p:120-134}.  

We apply the aLSI described in Section \ref{sec:plsi}. We use a metric Multidimensional scaling to obtain a low dimensional representation of the genes, which is shown in Figure \ref{figure:rep}. Also, the projections using the Pearson correlation and the Euclidean distance are shown. The Euclidean distance and the Pearson correlation do not show any cluster structure helpful to identify groups of genes involved in different cancers. However, the proposed aLSI  is able to do so. 

In order to interpret such groups we estimate the mixture model described in Section \ref{sec:plsi} with 14 groups. Each gene is assigned to a cluster by taking

$$\mbox{class of } gene_i =  \arg \max_{c_i} p(c_i| gene_i). $$
The conditional probabilities $p(c_i| gene_i)$  can be interpreted in this context as fuzzy membership degrees. In Table \ref{table.10genes} we show the 10 genes with the highest probability of each cluster. In Table \ref{table:groups}, we show the cross frequencies of the genes in the different types of cancers and clusters. Note that the same gene might belong to different cancer groups simultaneously, therefore the correspondence clusters-cancer types should not be necessarily one to one. 

Some interesting conclusions show up when the Table \ref{table:groups} is interpreted. BREAST, COLON, MELANONA, NSCLS and RENAL cancers seem to be associated to single clusters. The cancers K562A-repro and K562B-repro appear clearly together in the same group (group 9), which also occurs with cancers  MCF7A-repro and MCF7D-repro. 
Apart from the interpretability of the groups in terms of types of cancers, Table \ref{table:groups} also helps to identify similarities between types of cancer. Similar patterns between cancers across the clusters (similar rows) can be associated to similar types of cancer. The previously mentioned case of the K562A-repro and K562B-repro types is a clear example. A graphic illustration of these results can be observed in Figure \ref{cancer.map}, which shows a Sammon mapping of the 14 latent genetic classes (types of cancer) using the results from Table \ref{table:groups}.

\begin{figure}[t!]
 \begin{center}
\includegraphics[height=16cm,]{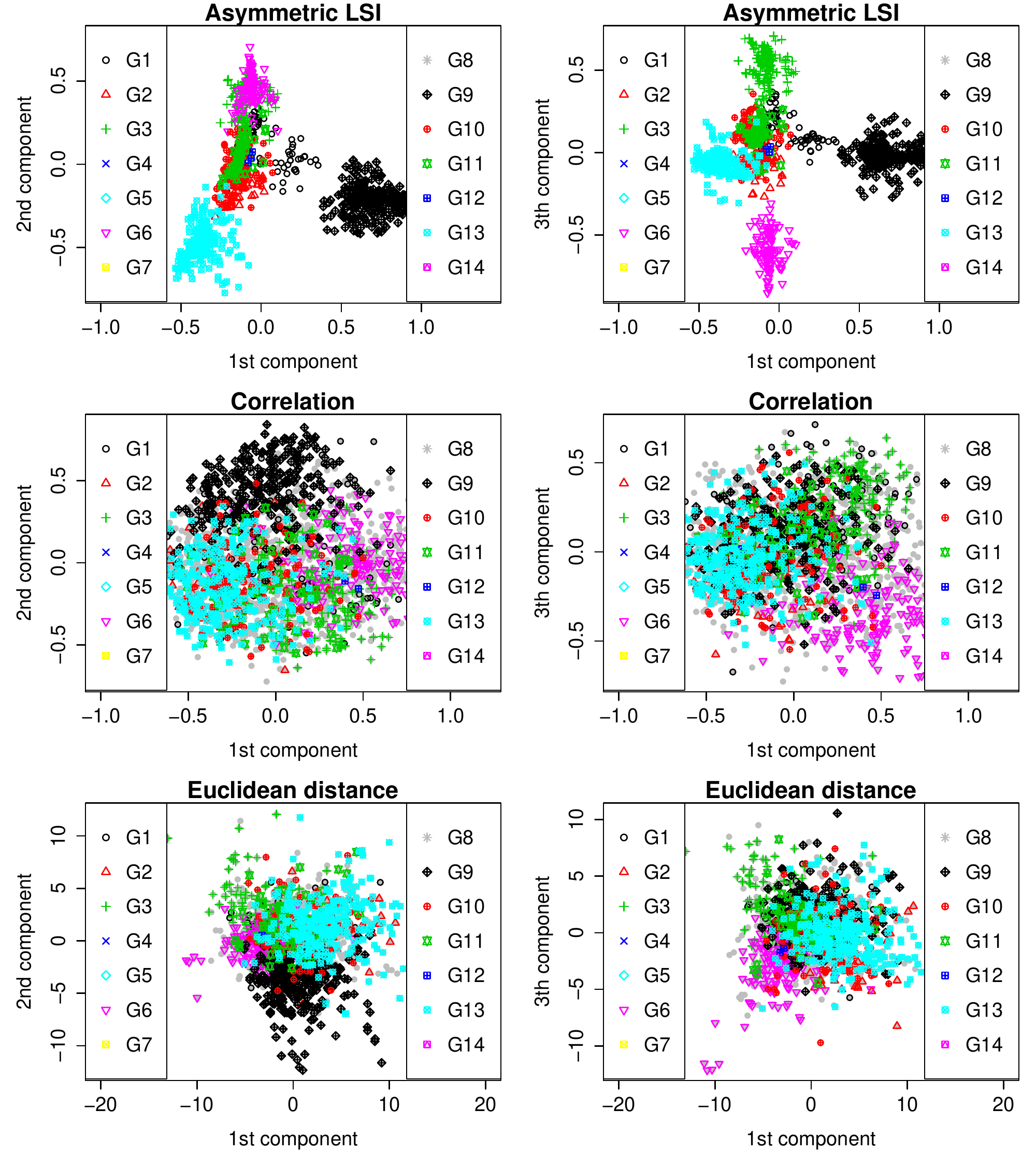}
\caption{ 
Multidimensional scaling projections (1st, 2nd, 3th) using the similarity matrix produced by the aLSI, the Pearson correlation and the Euclidean distance. The groups colouring correspond to the membership of the genes to the different groups of cancer: G1 (BREAST), G2 (CNS), G3 (COLON), G4 (K562A-repro), G5 (K562B-repro), G6 (LEUKEMIA), G7 (MCF7A-repro), G8 (MCF7D-repro), G9 (MELANOMA), G10 (NSCLC), G11 (OVARIAN), G2 (PROSTATE), G13 (RENAL), G14 (UNKNOWN). }\label{figure:rep}
  \end{center}
\end{figure}

\begin{table}[t!]
\centering
\tabcolsep=0.11cm 
\begin{tabular}{lcccccccccccccc}
  \hline
 & C1 & C2 & C3 & C4 & C5 & C6 & C7 & C8 & C9 & C10 & C11 & C12 & C13 & C14 \\ 
  \hline
BREAST &   0 &   1 &   0 &  30 &   0 & 301 &  24 &   1 &  11 &  41 &   8 &   2 &  14 &  29 \\ 
  CNS &   5 &  30 &  38 &  11 &  57 &  67 &   0 &  11 &   7 &   0 &  12 &   2 &   2 &   1 \\ 
  COLON &   0 &  15 &   5 & 287 &   0 &  17 &   0 &   2 &   5 &   6 &   0 &   4 &   7 &   0 \\ 
  K562A-repro &   1 &   1 &   0 &   0 &   0 &   0 &   0 &   0 &  62 &   0 &   0 &   2 &   0 &   0 \\ 
  K562B-repro &   0 &   2 &   0 &   1 &   0 &   0 &   0 &   0 &  64 &   0 &   0 &   0 &   1 &   0 \\ 
  LEUKEMIA &   0 &  21 &  13 &  33 &   0 &  48 &   1 & 147 &  62 &   4 &   1 &  56 &   2 &   0 \\ 
  MCF7A-repro &   0 &   0 &   0 &   1 &   0 &   1 &   0 &   0 &   1 &  45 &   0 &   0 &   0 &   0 \\ 
  MCF7D-repro &   0 &   2 &   0 &   1 &   0 &   5 &   0 &   0 &   0 &  36 &   0 &   0 &   1 &   0 \\ 
  MELANOMA &   1 &  24 &  39 &  21 &   0 &  64 &   0 &   8 &   6 &   5 & 307 &   0 &   8 &   1 \\ 
  NSCLC &   0 & 259 &   5 &  25 &   0 &  80 &   0 &   1 &   3 &   3 &   0 &   1 &  43 &   0 \\ 
  OVARIAN &  86 &  36 &  29 &  42 &   2 &  36 &   0 &  12 &   7 &   3 &  12 &   1 &   5 &   0 \\ 
  PROSTATE &   4 &  11 &   4 &   3 &   0 &  11 &   0 &   0 &   5 &   0 &   2 &   0 &   1 &   0 \\ 
  RENAL &   0 &  47 & 377 &  17 &   0 &  87 &   0 &   6 &   6 &   1 &   0 &   0 &   1 &   0 \\ 
  UNKNOWN &   5 &   7 &   2 &   2 &   0 &   5 &   0 &   0 &   0 &   0 &   0 &   0 &   0 &   0 \\ 
   \hline
\end{tabular}\setlength{\tabcolsep}{1em}
\caption{Correspondence between the 14 latent genetic estimated clusters and the genes membership to the different types of cancers.}\label{table:groups}
\end{table}

\begin{figure}[t!]
 \begin{center}
\includegraphics[height=10cm,]{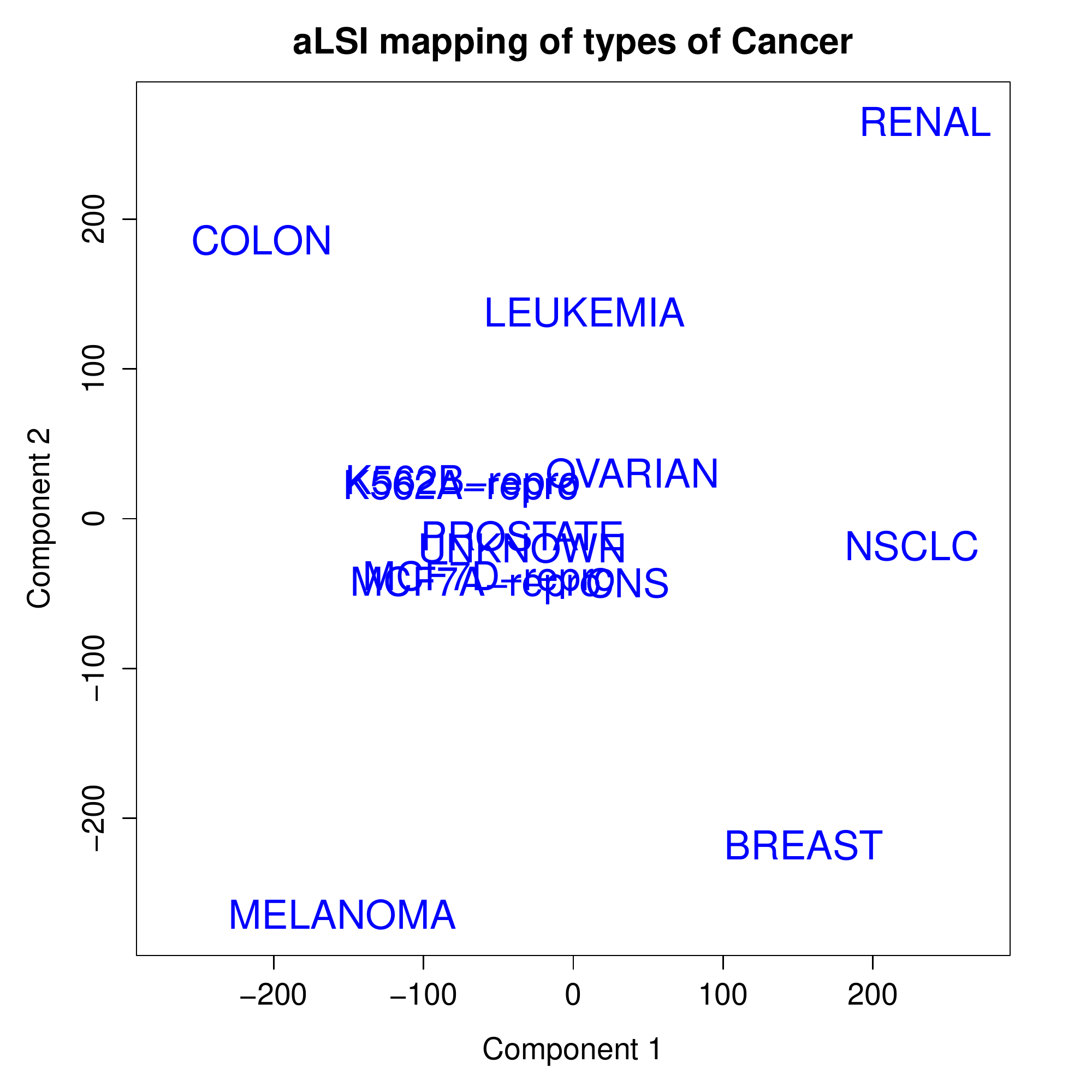}
\caption{ Sammon mapping of the 14 types of cancers using the results from Table \ref{table:groups} }\label{cancer.map}
  \end{center}
\end{figure}

\section{Conclusions}\label{sec5}

In this paper we have proposed a new approach to visualize gene expression experiments. The key idea is to use an asymmetric similarity for the genes, which is used within the latent semantic indexing context, to obtain latent genetic classes or groups of genes which are similar in their expression patterns. We provide both, a Euclidean representation of the genes, which is able to illustrate the different genetic patterns of expression in the data set, and the probabilities of membership of each gene to those classes. The proposed method has been used to analyse the Human Cancer dataset obtaining new and valuable information that remains unadvertised to classical similarity measures like the Pearson's correlation and the Euclidean distance.

This work leads to a wide variety of future analysis. On the most theoretical and methodological side, the study of the geometrical properties of the matrices $\textbf{K}_1$ and $\textbf{K}_2$ and of further combination procedures are of interest. For instance, we aim to explore the Geometric and Harmonic weighted means given by  

$$\mathcal{F}^t_{geometric}(\textbf{K}_1, \textbf{K}_2) =  \textbf{K}_1^{1/2}(\textbf{K}_1^{-1/2}\textbf{K}_2\textbf{K}_1^{-1/2})^{t}\textbf{K}_1^{1/2},$$
$$\mathcal{F}^t_{harmonic}(\textbf{K}_1, \textbf{K}_2) = (t\textbf{K}_1^{-1}+ (1-t) \textbf{K}_2^{-1})^{-1},$$
for $t \in [0,1]$ and to study their effects in the final genes representation.

 In addition, although we have presented a method in which the sources of asymmetry for the genes similarity are merged into a symmetric matrix, it is our plan to investigate the potential combinations of our approach with previously developed asymmetric multidimensional scaling techniques \citep{Chino2012}. Also new ways to embed prior knowledge into the matrix $\textbf{W}$ will be the focus of further study, which we envision will have a large impact for practitioners: in this work we only have considered the experiments labelling to obtain a measure of association for the genes. However, in the future it is our aim to consider gene ontologies and other topological measures of biological networks, like Protein-Protein interaction networks to improve the final gene mapping and the interpretation of the obtained gene semantic classes.

\appendix
\section{Appendix}

\begin{proof}
(Proposition 1). To maximize $\textsf{G}_{\tau}(\textbf{K})$ we take partial the derivative for each $\textbf{S}_{ls}$. Then
\begin{equation}
\frac{\partial \textsf{G}_{\tau}[\textbf{K}]}{\partial (\textbf{K})_{sl}} =  2(\textbf{K}_{ls} - \gamma_1\mathcal{F}(\textbf{K}_1,\textbf{K}_2)_{ls}) +2\tau \left((\textbf{K})_{ls} -  \gamma_2(\textbf{W})_{ls}  \right) 
\end{equation}
for $s,l=1,\dots,m,$. Setting the previous partial derivatives to zero yields a linear system whose unique solution is a matrix $\textbf{K}$ whose elements are given by

\begin{equation}\label{eq:proof}
\textbf{K}^*=\gamma_1\frac{1}{\tau+1} \mathcal{F}(\textbf{K}_1,\textbf{K}_2) +\gamma_2\frac{\tau}{\tau+1}  \textbf{W} = \mathcal{F}(\textbf{K}_1,\textbf{K}_2) +\tau \textbf{W},  
\end{equation}
for $l,s=1,\dots,m,$. To check if $\textbf{K}$ is a maximum or a minimum we evaluate the Hessian matrix of $\textsf{G}_{\tau}[\textbf{S}]$ on $\textbf{K}$. Such
matrix is the $n \times n$ diagonal matrix

\begin{equation}
H(\textbf{K}^*) = 2 \cdot\left(
  \begin{array}{cccc}
   \tau+1 & 0 & \cdots & 0 \\
    0 & \tau+1 & \cdots & 0 \\
    \vdots & \ddots & \vdots & \vdots \\
    0 & 0 & \cdots & \tau+1 \\
\end{array} \right)
\end{equation}
which is positive definite for any $\tau>0$. Hence, (\ref{eq:proof}) is a minimum of  (\ref{eq:optSpenalized}) for any $\tau>0$.
\end{proof}

\begin{proposition}
 let $k_1$ and $k_2$ two positive semi-definite kernel functions and let $\phi_1$ and $\phi_2$ their underlying feature maps. Then $k=\lambda_1k_1+\lambda_2k_2$, with $\lambda_1,\,\lambda_2\geq 0$, is a positive semi-definite kernel with $\phi = [\sqrt{\lambda_1}\phi_1,\sqrt{\lambda_2}\phi_2]$ as a valid feature map.
\end{proposition}

\begin{proof}
(Proposition 2). We only need to show that  $k(\textbf{t},\textbf{t}') = \langle \phi(\textbf{t}),\phi(\textbf{t}') \rangle$ is satisfied for $k$ and $\phi$. In our case we have that 

\begin{eqnarray} \nonumber
\langle \phi(\textbf{t}),\phi(\textbf{t}') \rangle & =& \langle (\sqrt{\lambda_1}\textbf{t}_1,\sqrt{\lambda_2}\textbf{t}_2) ,(\sqrt{\lambda_1}\textbf{t}'_1,\sqrt{\lambda_2}\textbf{t}'_2)  \rangle \\ \nonumber
& = &  \lambda_1\langle \phi_1(\textbf{t}),\phi_1(\textbf{t}') \rangle    +  \lambda_2 \langle \phi_2(\textbf{t}),\phi_2(\textbf{t}') \rangle  \\ \nonumber
& = & \lambda k_1(\textbf{t},\textbf{t}') + \lambda_2 k_2(\textbf{t},\textbf{t}') \\ \nonumber
 & = & k(\textbf{t},\textbf{t}'), \\ \nonumber
\end{eqnarray}
which shows that the proposition holds.
\end{proof}



\textbf{Acknowledgments}
We thank the support of the Spanish Grant Nos. MEC-2007/04438/00 and DGULM-2008/00059/00. We also thank Georges E. Janssens for hisl elpfull comments on the manuscript.

\bibliographystyle{apalike}
\bibliography{javbib}

\end{document}